\documentclass[11pt,letterpaper]{article}
\usepackage[T1]{fontenc}
\usepackage[margin=1in]{geometry}
\usepackage{fullpage}
\usepackage{url}
 \usepackage{amsmath}
\usepackage{amsthm}
\usepackage{amsfonts}
\usepackage{authblk}
\usepackage{hyperref}

\usepackage{xcolor}

\usepackage{graphicx}
\usepackage{caption}
\usepackage{subcaption}
\usepackage{comment}

\newtheorem{theorem}{Theorem}

\def\cH{\mathcal{H}}

\def\err{\mathrm{err}}
\def\errhot{\mathrm{errhot}}

 \date{}

	\title{Count-min sketch with variable number of hash functions: an experimental study}
	\author[1]{\'Eric Fusy}
\author[1]{Gregory Kucherov}
\affil[1]{LIGM, CNRS, Univ. Gustave Eiffel, Marne-la-Vall\'ee, France \texttt{\{Eric.Fusy|Gregory.Kucherov\}@univ-eiffel.fr}}

\begin{document}

	\maketitle
	
	\begin{abstract}
		Conservative Count-Min, a stronger version  of the popular Count-Min sketch [Cormode, Muthukrishnan 2005], %\cite{cormode_improved_2005}, 
		is an online-maintained hashing-based sketch summarizing element frequency information of a stream. Although several works attempted to analyze the error of conservative Count-Min, its behavior remains poorly understood. In [Fusy, Kucherov 2022], %\cite{FusyKucherovCIAC23}, 
		we demonstrated that under the uniform distribution of input elements, the error of conservative Count-Min follows two distinct regimes depending on its load factor. 
		
		In this work, we present a series of results providing new insights into the behavior of conservative Count-Min. 
		Our contribution is twofold. On one hand, we provide a detailed experimental analysis of Count-Min sketch in different regimes and under several representative probability distributions of input elements. On the other hand, we demonstrate improvements that can be made by assigning a variable number of hash functions to different elements. This includes, in particular, reduced space of the data structure while still supporting a small error. 
	\end{abstract}
	
% \newpage

	\section{Introduction}
In most general terms, 
	\textit{Count-Min sketch} is a data structure for representing an associative array of numbers indexed by elements (keys) drawn from a large universe, where the array is provided through a stream of (key,value) updates so that the current value associated to a key is the sum of all previous updates of this key. Perhaps the most common setting for applying Count-Min, that we focus on in this paper, is the \textit{counting} setting where all update values are $+1$. In this case, the value of a key is its \textit{count} telling how many times this key has appeared in the stream. In other words, Count-Min can be seen as representing a \textit{multiset}, that is a mapping of a subset of keys to non-negative integers. With this latter interpretation in mind, each update will be called \textit{insertion}. The main supported query of Count-Min is retrieving the count of a given key, and the returned estimate may not be exact, but can only overestimate the true count. 
	
	The counting version of Count-Min is applied to different practical problems related to data stream mining and data summarization. One example is tracking frequent items (\textit{heavy hitters}) in streams  \cite{liu2011methods,charikar2004finding,cormode2008finding}. It occurs in network traffic monitoring \cite{DBLP:conf/sigcomm/EstanV02}, optimization of cache usage  \cite{DBLP:conf/iccnc/EinzigerF15}. It also occurs in non-streaming big data applications, e.g. in bioinformatics \cite{mohamadi2017ntcard,behera2018kmerestimate,shibuya_set-min_2020}. 
	
	Count-Min relies on hash functions but, unlike classic hash tables, does not store elements but only count information (hence the term \textit{sketch}). It was proposed in \cite{cormode_improved_2005}, however a very similar data structure was proposed earlier in \cite{DBLP:conf/sigmod/CohenM03} under the name \textit{Spectral Bloom filter}. The latter, in turn, is closely related to \textit{Counting Bloom filters} \cite{FanEtAl00}. In this work, we adopt the definition of \cite{DBLP:conf/sigmod/CohenM03} but still call it Count-Min to be consistent with the name commonly adopted in the literature. A survey on Count-Min can be found e.g. in \cite{DBLP:reference/db/Cormode18}. 
	
	In this paper, we study a stronger version of Count-Min called \textit{conservative}. This modification of Count-Min was introduced in \cite{DBLP:conf/sigcomm/EstanV02} under the name \textit{conservative update}, see \cite{DBLP:reference/db/Cormode18}. It was also discussed in \cite{DBLP:conf/sigmod/CohenM03} under the name \textit{minimal increase}. Conservative Count-Min provides strictly tighter count estimates using the same memory and thus strictly outperforms the original version. The price to pay is the impossibility to deal with deletions (negative updates), whereas the original Count-Min can handle deletions as well, provided that the cumulative counts remain non-negative (condition known as {\em strict turnstile model} \cite{liu2011methods}). 
	
	Analysis of error of conservative Count-Min is a difficult problem having direct consequences on practical applications. Below in Section~\ref{sec:prior-work} we survey known related results in more details. In our previous work \cite{FusyKucherovCIAC23}, we approached this problem through the relationship with \textit{random hypergraphs}. We proved, in particular, that if the elements represented in the data structure are uniformly distributed in the input, the error follows two different regimes depending on the \textit{peelability} property of the underlying \textit{hash hypergraph}. While properties of random hypergraphs have been known to be crucially related to some data  structures (see Section~\ref{sec:hash}), this had not been known for Count-Min. 
	
	Starting out from these results, in this paper we extend and strengthen this analysis in several ways, providing experimental demonstrations in support of our claims. Our first goal is to provide a fine analysis of the ``anatomy'' of conservative Count-Min, describing its behavior in different regimes. Our main novel contribution is the demonstration that assigning different number of hash functions to different elements can significantly improve the error, and, as a consequence, lead to memory saving. Another major extension concerns the probability distribution of input elements: here we study non-uniform distributions as well, in particular step distribution and Zipf's distribution, and analyze the behavior of Count-Min for these distributions. This analysis is important not only because non-uniform distributions commonly occur in practice, but also because this provides important insights 
	 for the major application of Count-Min: detection of most frequent elements (sometimes called  
	%for the 
	\textit{heavy hitters} problem \cite{liu2011methods,charikar2004finding,cormode2008finding}). In particular, we consider the ``small memory regime'' (\textit{supercritical}, in our terminology) when the number of distinct represented elements is considerably larger than the size of the data structure, and analyse conditions under which most frequent elements are evaluated with negligible error. This has direct applications to the frequent elements problem. 
	
	To conclude the introduction, we note that the experimental character of our analysis does not restrict the generality of our results that hold for a wide range of parameters. This follows from the general nature of tested hypotheses, as well as from theoretical justifications based on previous works. 
	
	\section{Background and related work}
	
	\subsection{Conservative Count-Min: definitions}
	\label{countmin-and-CU}
	A Count-Min sketch is a counter array $A$ of size $n$ together with a set of hash functions mapping elements (keys) of a given universe $U$ to $[1..n]$. In this work, each element $e\in U$ can in general be assigned a different number $k_e$ of hash functions. Hash functions are assumed fully random, therefore we assume w.l.o.g. that an element $e$ is assigned hash functions $h_1,\ldots,h_{k_e}$. 
	
	At initialization, counters $A[i]$ are set to $0$. When processing an insertion of an input element $e$, basic Count-Min  increments by $1$ each counter $A[h_i(e)]$, $1\leq i\leq k_e$. The conservative version of Count-Min increments by $1$  only the smallest of all $A[h_i(e)]$. That is, $A[h_i(e)]$ is incremented by $1$ if and only if $A[h_i(e)]=\min_{1\leq j\leq k_e}\{A[h_j(e)]\}$ and is left unchanged otherwise. 
	
	In both versions, the \textit{estimate} of the number of occurrences of a queried element $e$ is computed by $c(e)=\min_{1\leq i\leq k_e}\{A[h_i(e)]\}$. It is easily seen that for any input sequence of elements, the estimate computed by original Count-Min is greater than or equal to the one computed by the conservative version. 
This follows from the observation that on the same input, an entry of counter array $A$ under conservative update can never get larger than the same entry under Count-Min. 
	
	In this work, we study the conservative version of Count-Min. Let $H$ denote a selection of hash functions $H=\{h_1,h_2,\ldots\}$. Consider an input sequence $I$ of $N$ insertions and let $E$ be the set of distinct elements appearing   
	in $I$. The \textit{relative error} of an element $e$ is defined by $\err_{H,I}(e) = (c(e) - occ(e))/occ(e)$, where $occ(e)$ is the number of occurrences of $e$ in the input. The \textit{combined error} is an average error over all elements in $I$ weighted by the number of occurrences, i.e. 
	$$\err_{H,I} = \frac{1}{N}\sum_{e\in E}occ(e) \cdot\err(e)=\frac{1}{N}\sum_{e\in E}(c(e) - occ(e)).$$
	
	We assume that $I$ is an i.i.d. random sequence drawn from a probability distribution on a set of elements $E\subseteq U$.  A key parameter is the size of $E$ relative to the size $n$ of $A$. By analogy to hash tables, $\lambda = |E|/n$ is called the \textit{load factor}, or simply the \textit{load}. 
	
	\subsection{Analysis of conservative Count-Min: prior works} 
	\label{sec:prior-work}
	Motivated by applications to traffic monitoring, \cite{DBLP:conf/teletraffic/BianchiDLS12} was probably the first work devoted to the analysis of conservative Count-Min in the counting setting. Their model assumed that all $\binom{n}{k}$ counter combinations are equally likely, where $k$ hash functions are applied to each element. This implies the regime when $|E| \gg n$. The focus of \cite{DBLP:conf/teletraffic/BianchiDLS12} was on the analysis of the \textit{growth rate} of counters, i.e. the average number of counter increments per insertion, using a technique based on Markov chains and differential equations. % We will come back to the results of \cite{DBLP:conf/teletraffic/BianchiDLS12} later on in our work. 
	%
%	It was observed, in particular, that under the assumption of large set of keys, two hash functions produce the smallest growth rate and therefore the smallest error. While the analysis of \cite{DBLP:conf/teletraffic/BianchiDLS12} was done in the assumption of uniform distribution of keys, it was observed that for very skewed distributions, such as Zipf's distribution, only a few most frequent keys have no error whereas all other keys are evaluated to about the same value irrespective of their count. Interestingly, this value appears to be the largest in the case of uniform distribution, which implies that the case of the uniform distribution has direct implications the to non-uniform case. 
		% Their  motivation for analysing the uniform distribution is to provide a conjectural frequency threshold (for given input stream length and number of counters), beyond which the count approximation is always accurate, whatever the distribution for input keys. 
		%
	Another approach proposed in \cite{DBLP:conf/iccnc/EinzigerF15}      
	simulates a conservative Count-Min sketch  by a hierarchy of ordinary Bloom filters. 
	Obtained error bounds are not explicit but are 
	expressed via a recursive relation based on false positive rates of corresponding Bloom filters. 
	% similar to the Count-Min one (error bounded by $\varepsilon \lVert \boldsymbol{a} \rVert_{1}$), and the probability bounds are not explicit. 
	% The talk about Zipfian distributions but only in the experimental part, nothing about that in the analysis. 

	Recent works \cite{benmazziane:hal-03613957,benmazziane22b} propose an analytical approach for computing error bounds depending on element probabilities assumed independent but not necessarily uniform, in particular leading to improved precision bounds  for detecting heavy hitters.  
	However the efficiency of this technique is more limited when all element probabilities are small.    
	In particular, if the input distribution is uniform,
	their approach does not bring out any improvement over the general bounds known for original Count-Min.  
	
	In our recent work \cite{FusyKucherovCIAC23}, we proposed an analysis of conservative Count-Min based on its relationship with {random hypergraphs}. We summarize the main results of this work below in Section~\ref{sec:peelability}.
	
	\subsection{Hash hypergraph} 
	\label{sec:hash}
		Many hashing-based data structures are naturally associated with hash hypergraphs so that hypergraph properties are directly related to the proper functioning of the data structure. This is the case with Cuckoo hashing \cite{pagh2004cuckoo} and Cuckoo filters \cite{10.1145/2674005.2674994}, Minimal Perfect Hash Functions and Static Functions \cite{majewski1996family}, Invertible Bloom Lookup Tables \cite{goodrich2011invertible}, and some others. \cite{DBLP:phd/dnb/Walzer20} provides an extended study of relationships between hash hypergraphs and some of those data structures. 
		
	A Count-Min sketch 
	is associated with a \textit{hash hypergraph}  $H=(V,E)$ where $V=\{1..n\}$
	and $E=\{\{h_1(e),...h_{k_e}(e)\}\}$ over all distinct input elements $e$. We use notation $\cH_{n,m}$ for hypergraphs with $n$ vertices and $m$ edges, and $\cH_{n,m}^{k}$ for $k$-uniform such hypergraphs, where all edges have cardinality $k$. In the latter case, since our hash functions are assumed fully random, a hash hypergraph is a $k$-uniform Erd\H{o}s-R\'enyi random hypergraph. 
	
	As inserted elements are assumed to be drawn from a random distribution, it is convenient to look at the functioning of a Count-Min sketch as a stochastic process on the associated hash hypergraph % (called \textit{CU-process} 
	\cite{FusyKucherovCIAC23}.  Each vertex holds a counter initially set to zero, and therefore each edge is associated with a set of counters held by the corresponding vertices. Inserting an element consists in incrementing the minimal counters of the corresponding edge, and retrieving the estimate of an element returns the minimum value among the counters of the corresponding edge. From now on in our presentation, we will interchangeably speak of distinct elements and edges of the associated hash hypergraph, as well as of counters and vertices. 
	% In the rest of the paper, for the sake of consistency, we mainly use graph terminology and speak about edges and nodes (node values) instead of speaking of keys and counters (counter values). 
	Thus, we will call the \textit{vertex value} the value of the corresponding counter, and the \textit{edge value} the estimate of the corresponding element. Also, we will speak about the \textit{load} of a hypergraph understood as the density $|E|/|V|$. 
	
	\subsection{Hypergraph peelability and phase transition of error} 
	\label{sec:peelability}
	A hypergraph $H=(V,E)$ is called \textit{peelable} if iterating the following step starting from $H$ results in the empty graph: if the graph has a vertex of degree $1$ or $0$, delete this vertex together with the incident edge (if any). Like many other properties of random hypergraphs, peelability undergoes a phase transition. Consider the  Erd\H{o}s-R\'enyi $k$-uniform hypergraph model where graphs are drawn from $\cH_{n,m}^{k}$ uniformly at random. 
	It is shown in~\cite{molloy2005cores} that a phase transition occurs at a (computable) peelability threshold $\lambda_k$:  
	a random graph from $\cH_{n,\lambda n}^{k}$ is with high probability (w.h.p.) peelable if $\lambda < \lambda_k$, and w.h.p. non-peelable if $\lambda>\lambda_k$.    
	The first values are $\lambda_2=0.5$, $\lambda_3\approx 0.818$, $\lambda_4\approx 0.772$, etc., $\lambda_3$ being the largest. Note that the case  $k=2$ makes an exception to peelability: for $\lambda < \lambda_2$, a negligible fraction of vertices remain after peeling. 
	
	Peelability is known to be directly relevant to certain constructions of Minimal Perfect Hash Functions  \cite{majewski1996family} as well as to the proper functioning of Invertible Bloom filters \cite{goodrich2011invertible}. In \cite{FusyKucherovCIAC23}, we proved 
	that it is relevant to Count-Min as well. 
	\begin{theorem}[\cite{FusyKucherovCIAC23}]\label{CIAC-subcritical}
		Consider a  conservative Count-Min where each element is hashed using $k$ random hash functions. Assume that the input $I$ of length $N$ is drawn from a {\em uniform} distribution on a set $E\subseteq U$ of elements and let $\lambda = |E|/n$, where $n$ is the number of counters. If $\lambda < \lambda_k$, then for a randomly chosen element $e$, the error $(c(e)-occ(e))/occ(e)$ is $o(1)$ w.h.p. when both $n$ and $N/n$ grow. 
	\end{theorem}
	In the complementary regime $\lambda > \lambda_k$, we showed in \cite{FusyKucherovCIAC23}, under some additional assumptions, that $\err_{H,I}$ is $\Theta(1)$. Thus, the peelability threshold for random hash hypergraphs corresponds to phase transition in the error produced by conservative Count-Min for uniform distribution of input. We call regimes $\lambda < \lambda_k$ and $\lambda > \lambda_k$  \textit{subcritical} and \textit{supercritical}, respectively. 
	
	\subsection{Variable number of hash functions: mixed hypergraphs}\label{sec:variable} 	The best peelability threshold $\lambda_3\approx 0.818$ can be improved in at least two different ways. One way is to use a carefully defined class of hash functions which replace uniform sampling of $k$-edges  by a specific non-uniform sampling. Thus, \cite{DBLP:conf/esa/DietzfelbingerW19} showed that the peelability threshold can be increased to $\approx 0.918$ for $k=3$ and up to $\approx 0.999$ for larger $k$'s if a special class of hypergraphs is used. 
	
	Another somewhat surprising idea, that we apply in this paper, is to apply a different number of hash functions to differents elements, that is to consider non-uniform hypergraphs. 
	Following \cite{DBLP:conf/icalp/DietzfelbingerGMMPR10},  \cite{DBLP:conf/sofsem/Rink13} showed that non-uniform hypergraphs may have a larger peelability threshold than uniform ones. 
	More precisely, \cite{DBLP:conf/sofsem/Rink13} showed that \textit{mixed hypergraphs} with two types of edges of different cardinalities, each constituting a constant fraction of all edges, may have a larger peelability threshold: for example, hypergraphs with a fraction of $\approx 0.887$ of edges of cardinality $3$ and the remaining edges of cardinality $21$ have the peelability threshold $\approx 0.920$, larger than the best threshold $0.818$ achieved by uniform hypergraphs. 
	% cite Michael G Luby, Michael Mitzenmacher, Mohammad Amin Shokrollahi, Daniel A Spielman, Efficient erasure correcting codes, IEEE 2001  ?
	We adopt the notation of \cite{DBLP:conf/sofsem/Rink13} for mixed hypergraphs: by writing $k=(k_1,k_2)$ we express that the hypergraph contains edges of cardinality $k_1$ and $k_2$, and $k=(k_1,k_2;\alpha)$ specifies in addition that the fraction of $k_1$-edges is $\alpha$. 
	
	The idea of using different number of hash functions for different elements has also appeared in data structures design. \cite{DBLP:conf/isit/BruckGJ06} proposed \textit{weighted Bloom filters} which apply a different number of hash functions depending on the frequency with which elements are queried and on probabilities for elements to  belong to the set. 
	The idea is to assign more hash functions to elements with higher query frequency and to those with smaller occurrence probability.  
	It is shown that this leads to a reduced false positive probability, where the latter is defined to be weighted by query frequencies. This idea was further refined in \cite{DBLP:conf/dasfaa/WangJDZZ15}, and then further in \cite{DBLP:journals/corr/abs-2205-14894}, under the name \textit{Daisy Bloom filter}.

\section{Results}
\label{sec:results}
\subsection{Uniform distribution}
\label{sec:uniform}
We start with the case where input elements are uniformly distributed, i.e. edges of the associated hash hypergraph have equal probabilities to be processed for updates. 

\vspace*{-3mm}
\subsubsection{Subcritical regime}
\label{sec:uniform-subcritical}
Theorem~\ref{CIAC-subcritical} in conjunction with the results of Section~\ref{sec:variable}
leads to the assumption that using a different number of hash functions for different elements one could ``extend'' the regime of $o(1)$ error of Count-Min sketch, which can be made into a rigorous statement (for simplicity we only give it with two different edge cardinalities).
\begin{theorem}[Extends~Theorem~\ref{CIAC-subcritical} to mixed cardinalities]
	\label{peel-mixed}
Consider a conservative Count-Min with $n$ counters. Assume that the input of length $N$ is drawn from a uniform distribution on $E\subseteq U$ and let $\lambda<\lambda_k$. Assume further that elements of $E$ are hashed according to a mixed hypergraph model $k=(k_1,k_2;\alpha)$. 
Let $c_k$ be the peelability ratio associated to $k$. Then, when $\lambda < c_k$, the error rate of a randomly chosen key is $o(1)$ w.h.p., as both $n$ and $N/n$ grow. 
\end{theorem}
\begin{proof}
	The crucial point is that the analysis of~\cite{FusyKucherovCIAC23} first establishes a deterministic result (Lemma~5 in the extended version) which can be stated as follows: the error rate of a peelable edge $e$ 
	in an arbitrary fixed hypergraph is $o(1)$, and the rate of convergence to $0$ is bounded in terms of the number $M_e$ of edges whose peeling has to precede the peeling of~$e$. 
	This ensures that, for any model of random hypergraphs, if a random edge $e$ in the hypergraph is peelable w.h.p. and if $M_e$ is $O(1)$ w.h.p., then the error rate
	of $e$ is $o(1)$. Then, the analysis performed in~\cite{DBLP:conf/icalp/DietzfelbingerGMMPR10} ensures that the model of random hypergraphs with mixed cardinalities has this property below the peelability threshold. 
\end{proof}

Figure~\ref{fig:uniform-subcritical} %confirms this assumption. It 
shows the average relative error as a function of the load factor for three types of hypergraphs: 2-uniform, 3-uniform and mixed hypergraph where a $0.885$ fraction of edges are of cardinality 3 and the remaining ones are of cardinality 14. 2-uniform and 3-uniform hypergraphs illustrate phase transitions at load factors approaching respectively $0.5$ and $\approx 0.818$, peelability thresholds for $2$-uniform and $3$-uniform hypergraphs respectively. It is clearly seen that the phase transition for the mixed hypergraphs occurs at a larger value approaching $\approx 0.898$ which is the peelability threshold for this class of hypergraphs \cite{DBLP:conf/sofsem/Rink13}. 
\begin{figure}[h]%{r}{0.5\textwidth}	
	\begin{center}
		\includegraphics[width=0.6\textwidth]{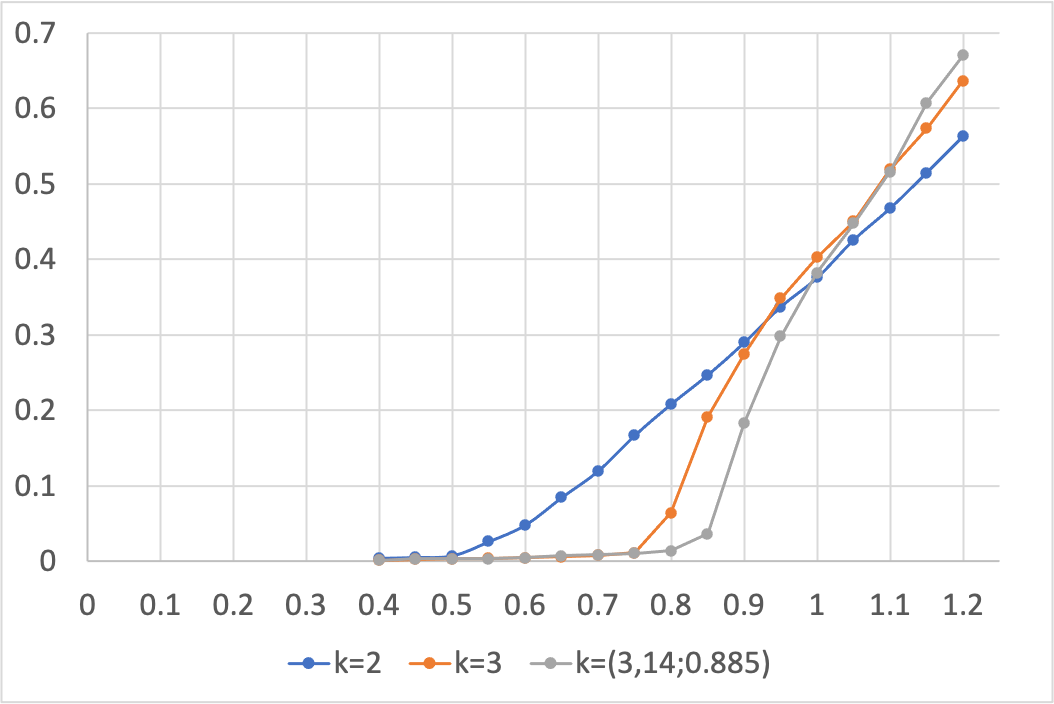}
	\end{center}
	\caption{$\err_{H,I}$ for small $\lambda=m/n$, for uniform distribution and different types of hypergraphs: 2-uniform, 3-uniform and (3,14)-mixed with a fraction of $0.885$ of 3-edges (parameters borrowed from \cite{DBLP:conf/sofsem/Rink13}). Data obtained for $n=1000$. The input size in each experiment is $5,000$ times the number of edges. Each average is taken over 10 random hypergraphs.}
	\label{fig:uniform-subcritical}
\end{figure}

While this result follows by combining results of \cite{DBLP:conf/sofsem/Rink13} and \cite{FusyKucherovCIAC23}, it has not been observed earlier and has an important practical consequence: \textit{using a varying number of hash functions in Count-Min sketch allows one to increase the load factor while keeping negligibly small error}. In particular, for the same input, this leads to space saving compared to the uniform case. 

Note that parameters $k=(3,14;0.885)$ are borrowed from \cite{DBLP:conf/sofsem/Rink13} in order to make sure that the phase transition corresponds to the peelability threshold obtained in \cite{DBLP:conf/sofsem/Rink13}. In practice, ``simpler'' parameters can be chosen, for example we found that $k=(2,5;0.5)$ produces essentially the same curve as $k=(3,14;0.885)$. 

\vspace*{-3mm}
\subsubsection{Supercritical regime}
\label{sec:uni-super} 
When the load factor becomes large (supercritical regime), the situation changes drastically. When the load factor just surpasses the threshold, some edges are still evaluated with small or zero error, whereas for the other edges, the error becomes large. This ``intermediate regime'' has been illustrated in \cite{FusyKucherovCIAC23}. Interestingly, edge values are distributed in this regime in a very peculiar way, concentrating around several values (see Figure~3 in \cite{FusyKucherovCIAC23} for illustration). These values must be explained by some graph structural patterns which remain to be elucidated. 

When the load factor goes even larger, the multi-level pattern of edge values disappears and all edge values become concentrated around the same value. 
We call this phenomenon \textit{saturation}. For example, for $k=3$ saturation occurs at around $\lambda = 6$ (data not shown). Under this regime, the hash hypergraph is dense enough so that its specific topology is likely to be irrelevant and the largest counter level ``percolates'' into all vertex counters. In other words, all counters grow at the same rate, without any of them ``lagging behind'' because of particular graph structural patterns (such as edges containing leaf vertices). 

Bianchi et al. \cite{DBLP:conf/teletraffic/BianchiDLS12} did their analysis under the assumption that each of $\binom{n}{k}$ edges is equally likely to be processed at each step. This emulates the situation where the load factor is very large and the hypergraph is saturated. The focus of \cite{DBLP:conf/teletraffic/BianchiDLS12} is on the \textit{growth rate} which is the expected number of counter increments when processing an edge. Obviously, this number varies between $1$ and $k$, and \cite{DBLP:conf/teletraffic/BianchiDLS12} establishes that larger values of $k$ imply larger growth rates. 
Note that the growth rate determines the slope of the linear dependence of $\err_{H,I}$ on $\lambda$. 
%% since the hypergraph is saturated and all counters are concentrated around the same value, 

The case $k=1$ has not been considered by \cite{DBLP:conf/teletraffic/BianchiDLS12}. In this case, the growth rate is trivially $1$, as each insertion increments exactly one counter. Furthermore, $\err_{H,I}$ can be easily inferred in this case, as the error of a given key is defined by the number of keys hashed to the same counter. The number of such keys is approximated by the Poisson distribution with parameter $\frac{m}{n}=\lambda$, with expected value $\lambda$. Since keys are uniformly distributed, $\err_{H,I}=\lambda$ is a good estimator of the average relative error. 

However, it is non-trivial to see how this compares to the error for $k=2$ in the case of moderate load factors. Figure~\ref{fig:uniform-supercritical} clarifies the situation. It shows that $k=1$ produces larger estimates for load factors until a certain point before going below the estimate for $k=2$. This ``intermediate regime'' roughly corresponds to the $n\ln n$ coupon collector bound, i.e. to the regime where a significant number of counters remain zero. In this case, even if the total sum of counters is smaller for $k=1$ than for $k=2$, it is more evenly distributed for the latter case, similar to the well-known \textit{power of choice} phenomenon  in resource allocation \cite{DBLP:journals/siamcomp/AzarBKU99}.  Since empty counters are irrelevant for the average relative error, this results in a smaller error for $k=2$ compared to $k=1$. 
Note that for $k=1$ we observe that the error curve fits with the diagonal line. This can be explained by the fact that, for any given key $x$, the number of other keys falling in the same counter as $x$ (which determines the error in the estimate for $x$) is classically approximated  by a Poisson random variable of parameter $\lambda=m/n$, of expected value $\lambda$.  

\begin{figure}[h]%{r}{0.5\textwidth}	
	\begin{center}
		\includegraphics[width=0.7\textwidth]{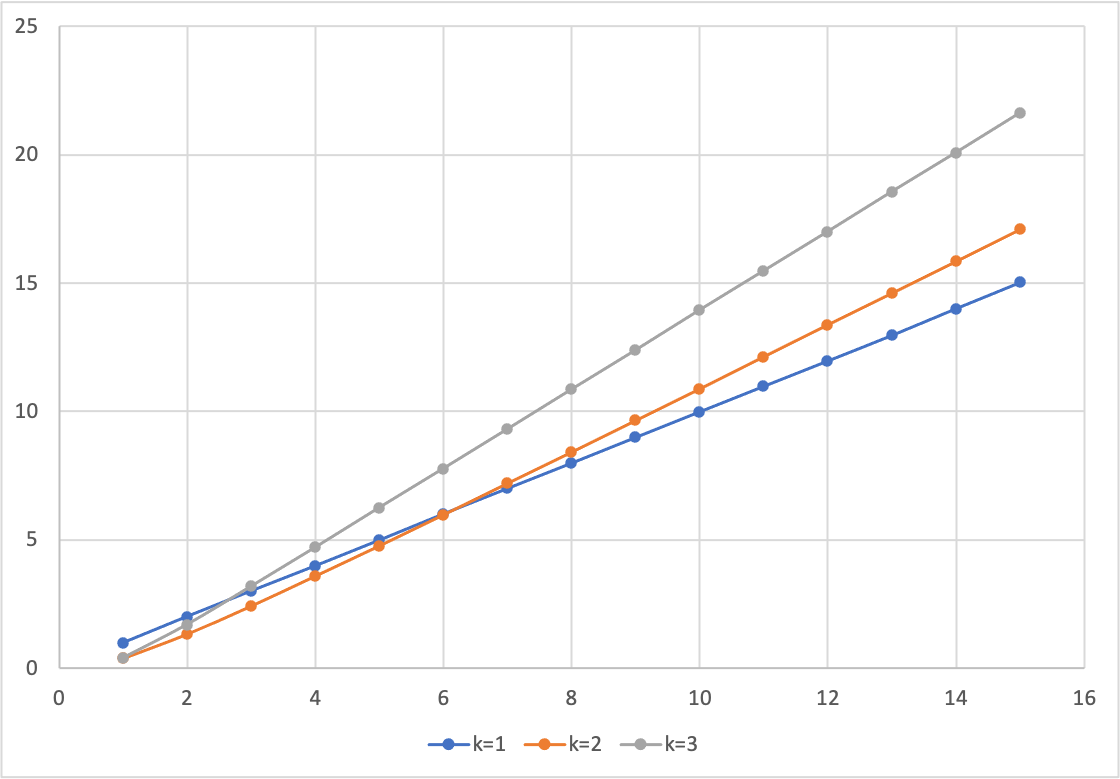}
	\end{center}
	\caption{$\err_{H,I}$ as a function of $\lambda=m/n$, for uniform distribution and supercritical regime. 
	}
	\label{fig:uniform-supercritical}
\end{figure}

Note that the configuration $k=1$ does not benefit from the advantage of the conservative update strategy over the regular Count-Min and has a limited interest from the practical viewpoint. It is however interesting to observe that $k=2$ still outperforms $k=1$ for limited~$\lambda$. 

\subsubsection{Mixed hypergraphs}
\label{sec:uniform-mixed}
We have just seen that for sufficiently large loads, the configuration $k=1$ results in the smallest error. Our next question is whether using a varying number of hash functions (mixed hypergraphs) can make the error even smaller, by analogy to the subcritical regime where it extends the range of load values supporting $o(1)$ error. At first glance, this question is not relevant, as the configuration $k=1$ has obviously the smallest possible sum of counters, since every insertion increments at least one counter and therefore $k=1$ seems to yield smallest possible estimates. 

This argument, however, applies to the regime when the hypergraph is saturated (all counters are hit), and $\lambda$ is large enough so that counters are concentrated. 
Note that as mentioned earlier,  counter values approximately follow a $\mathrm{Poisson}(k\lambda)$ distribution, which is concentrated around $k\lambda$ when $\lambda$ gets large. 
Perhaps surprisingly, it turns out that the error produced by $k=1$ can be made smaller for a large interval of $\lambda$ by using a varying number of hash functions, before $\lambda$ reaches the saturation point. 
Figure~\ref*{fig:uniform-supercritical-mixed} illustrates this phenomenon. 
It compares the error produced by the uniform case $k=1$ and a mixed configuration with $80\%$ of edges of cardinality $1$ and $20\%$ of edges of cardinality $3$. The latter produces a smaller error for values of $\lambda$ up to about $50$. Similar to the previous section, this is explained by the \textit{power of choice} effect: 3-edges ``smooth out'' counter values making the combined error smaller, 
 which in the not-yet-asymptotic regime $\lambda<50$ has a slightly more significant effect than 
the fact that the presence 3-edges makes the sum of counter values larger (as some insertions increment more than one counter). 

\begin{figure}[h]
	\centering
	\begin{subfigure}[b]{0.45\textwidth}
		\centering
		\includegraphics[scale=0.38]{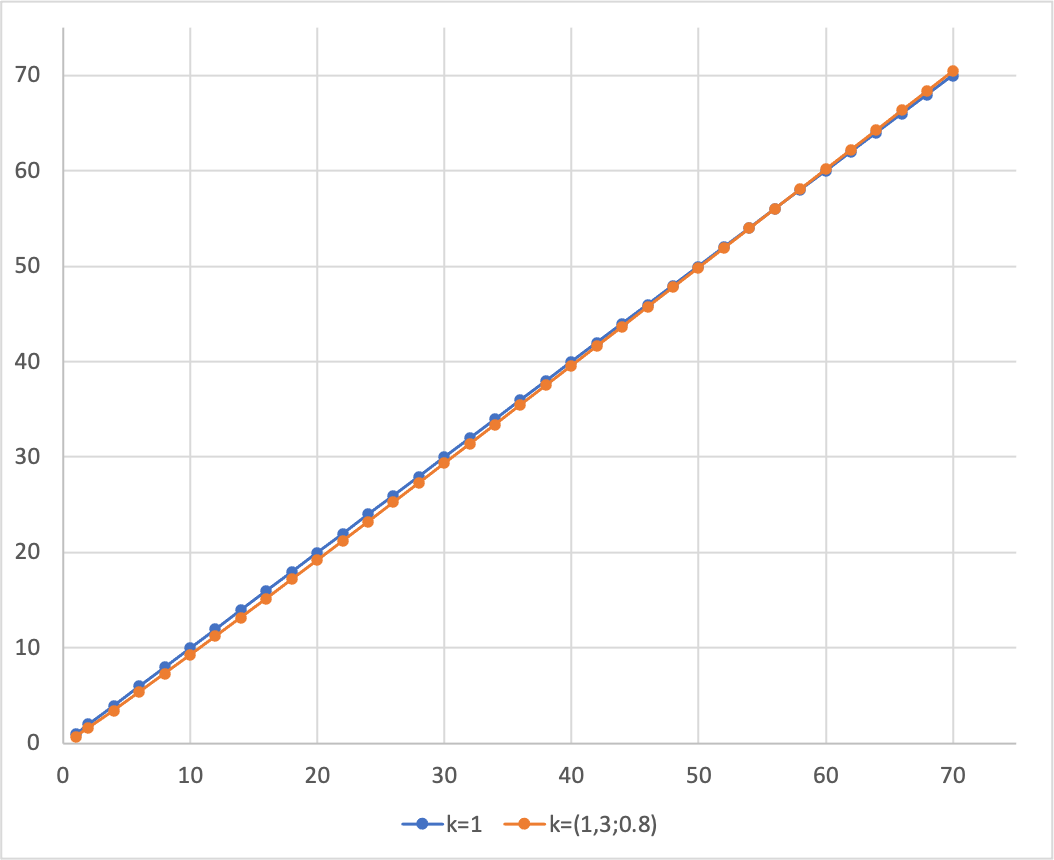}
		% \caption{Av erre}
		 %\label{fig:G20}
	\end{subfigure}
	\hspace*{15pt}
	\begin{subfigure}[b]{0.45\textwidth}
		\centering
%		% \includegraphics[width=\textwidth]{k3.jpg}
		\includegraphics[scale=0.38]{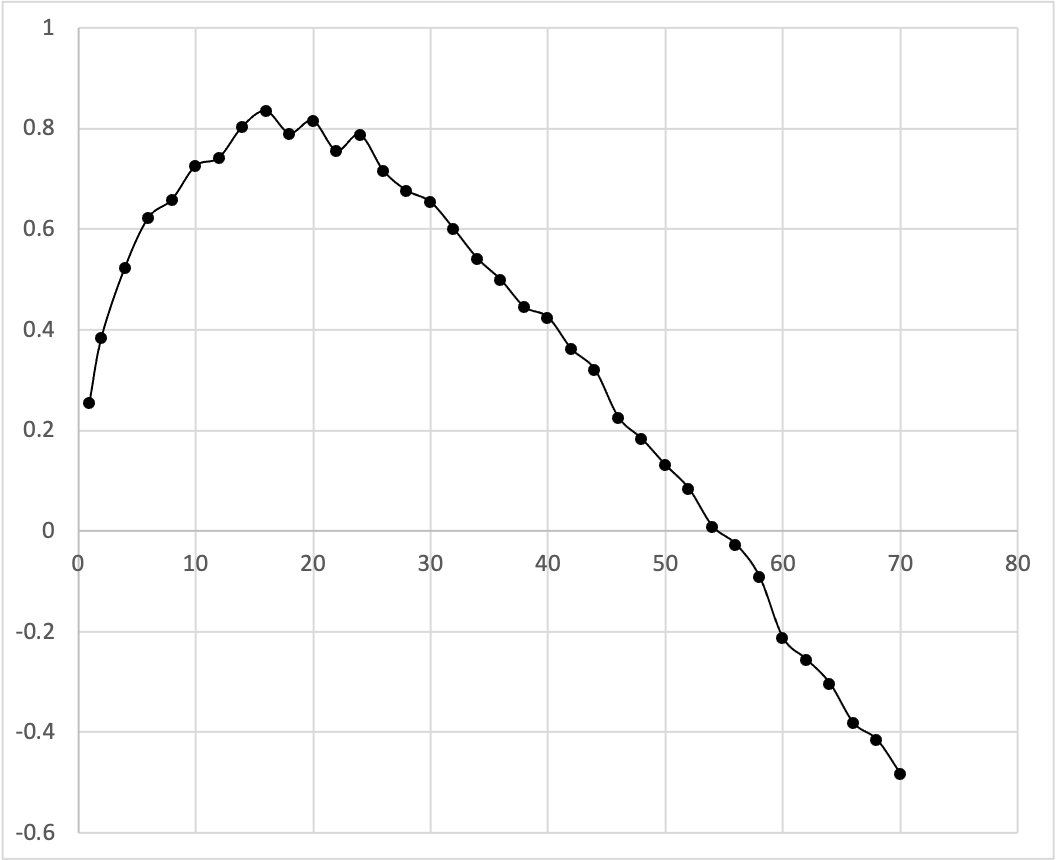}
%		% caption{$G=50$}
%		% \label{fig:G50}
	\end{subfigure}
	\caption{$\err_{H,I}$ as a function of $\lambda$, for uniform distribution and supercritical regime. Uniform configuration $k=1$ (blue curve) vs. mixed configuration $k=(1,3;0.8)$ (orange curve). % (80\% of edges) and $k=3$ (20\% of edges).
	For clarity, the difference between the former and the latter is shown in the right plot. }
	\label{fig:uniform-supercritical-mixed}
\end{figure}

\subsection{Step distribution}

The analysis of the behavior of conservative Count-Min under uniform distribution of input elements shows that in the supercritical regime, the error made by the sketch grows linearly with the load factor. 
This implies a limited practical utility of the sketch in this regime. 
On the other hand, in many practical situations, input elements are not occurring with the same frequency. This motivates the application of Count-Min to non-uniform distributions and, in particular, to detection and analysis of frequent elements in the input stream. One popular problem here is computation of \textit{heavy hitters}, where Count-Min sketch have been previously applied  \cite{cormode_improved_2005}. 

In this section, we focus on the simplest non-uniform distribution -- \textit{step distribution} -- in order to examine the behavior of Count-Min sketch in presence of elements with different frequencies. Our model is as follows. We assume that input elements are classified into two groups that we call \textit{hot} and \textit{cold}, where a hot element has a larger appearance probability than a cold one. Note that we assume that we have a prior knowledge on whether a given element belongs to hot or cold ones. This setting is similar to the one studied for Bloom filters augmented with prior membership and query probabilities \cite{DBLP:journals/corr/abs-2205-14894}. Note that our definition of $\err_{H,I}$ assumes that the query probability of an element and its appearance probability in the input are equal. 

We assume that the load factors of hot and cold elements are $\lambda_h$ and $\lambda_c$ respectively. That is, there are $\lambda_h n$ hot and $\lambda_c n$ cold edges in the hash hypergraph. 
$G>1$, called \textit{gap factor}, denotes the ratio between probabilities of a hot and a cold element respectively. Let $p_h$ (resp. $p_c$) denote the probability for an input element to be hot (resp. cold). Then 
$p_h/p_c=G\lambda_h/\lambda_c$, and since $p_h+p_c=1$, we have 
$$p_h=\frac{G\lambda_h}{\lambda_c+G\lambda_h}, ~~ p_c=\frac{\lambda_c}{\lambda_c+G\lambda_h}.$$
For example, if there are 10 times more distinct cold elements than hot ones ($\lambda_h/\lambda_c=0.1$) but each hot element is 10 times more frequent than a cold one ($G=10$), than we have about the same fraction of hot and cold elements in the input ($p_h=p_c=0.5$). 

In the rest of this section, we will be interested in the combined error of hot elements alone, denoted $\errhot_{H,I}$. If $E_h\subseteq E$ is the subset of hot elements, and $N_h$ is the total number of occurrences of hot elements in the input, then $\errhot_{H,I}$ is defined by
$$\errhot_{H,I} = \frac{1}{N_h}\sum_{e\in E_h}occ(e) \cdot\err(e)=\frac{1}{N_h}\sum_{e\in E_h}(c(e) - occ(e)).$$

\subsubsection{``Interaction'' of hot and cold elements}

A partition of elements into hot and cold induces the partition of the underlying hash hypergraph into two subgraphs that we call \textit{hot} and \textit{cold subgraphs} respectively. Since hot elements have larger counts, one might speculate that counters associated with hot edges are larger than counts of cold elements and therefore are not incremented by those. Then, $\errhot_{H,I}$ is entirely defined by the hot subgraph, considered under the uniform distribution of elements. In particular, $\errhot_{H,I}$ as a function of $\lambda_h$ should behave the same way as $\err_{H,I}$  for the uniform distribution (see Section~\ref{sec:uniform}). 

%\begin{figure}[h]%{r}{0.5\textwidth}	
%	\begin{center}
%		\includegraphics[width=0.7\textwidth]{step-transition-mixed.png}
%	\end{center}
%	\caption{Error of hot edges 2-uniform, 3-uniform and (2,5)-mixed hypergraphs. A hot edge (10\% of all edges) is 10 times more frequent in the input than a cold one. 
%	}
%	\label{fig:step-transition-mixed}
%\end{figure}
This conjecture, however, is not true in general. One reason is that there is a positive probability that all nodes of a cold edge are incident to hot edges as well. As a consequence, ``hot counters'' (i.e. those incident to hot edges) gain an additional increment due to cold edges, and the latter contribute to the overestimate of hot edge counts. 
% It follows that the above conjecture is not true, as illustrated by Figure~\ref{fig:step-transition-mixed}. 
%Consider the case $k=3$ (grey curve) in Figure~\ref{fig:step-transition-mixed}. If the conjecture was true, then the error would ``take off'' from 0 at $\lambda=\lambda_3 / \alpha \approx 8.18$, % ($=\lambda_3\cdot 10$),  whereas Figure~\ref{fig:step-transition-mixed} clearly shows that this happens at a much smaller value of $\lambda$. 
Figure~\ref{fig:G20} illustrates this point. It shows, for $k=3$, $\errhot_{H,I}$ as a function of $\lambda_h$ in presence of cold elements with $\lambda_c=5$, for the gap value $G=20$. For the purpose of comparison, the orange curve shows the error for the uniform distribution (as in Figure~\ref{fig:uniform-subcritical}), that is the error that hot elements would have if cold elements were not there. We clearly observe the contribution of cold elements to the error, even in the load interval below the peelability threshold. 

\begin{figure}[h]
	\centering
	\begin{subfigure}[b]{0.45\textwidth}
		\centering
		\includegraphics[scale=0.38]{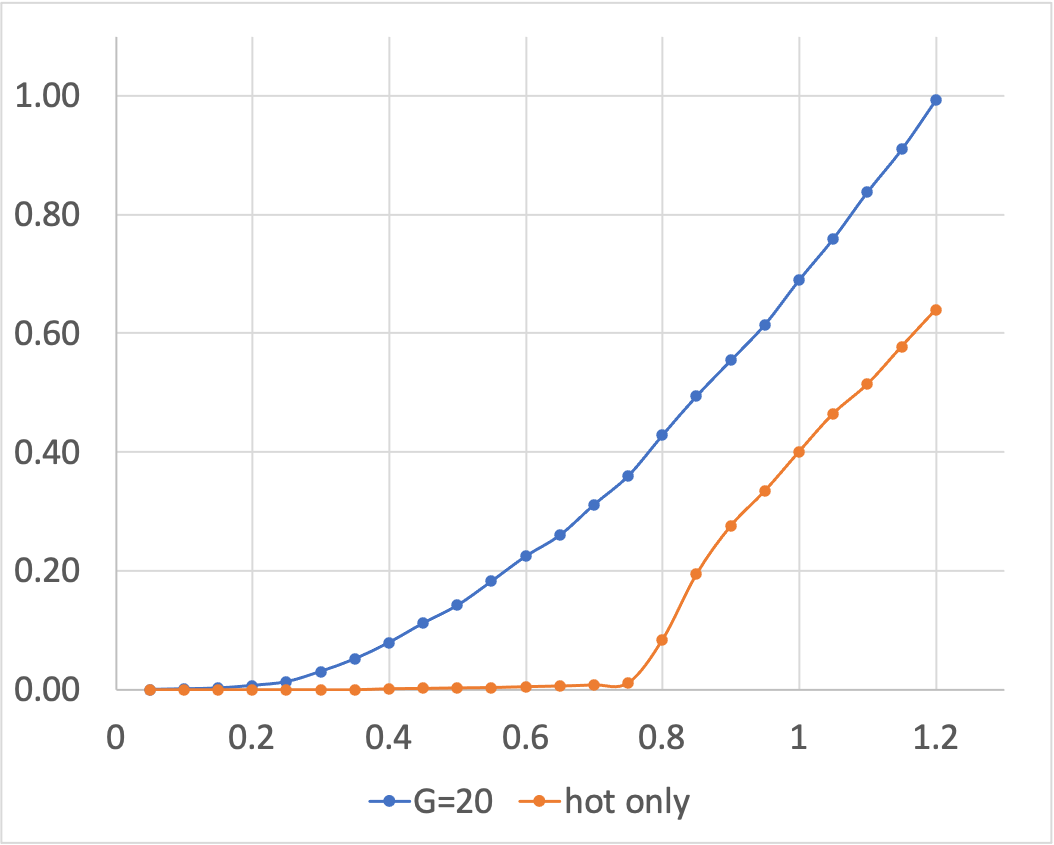}
		\caption{$G=20$}
		\label{fig:G20}
	\end{subfigure}
\hspace*{10pt}
	\begin{subfigure}[b]{0.45\textwidth}
		\centering
		\includegraphics[scale=0.38]{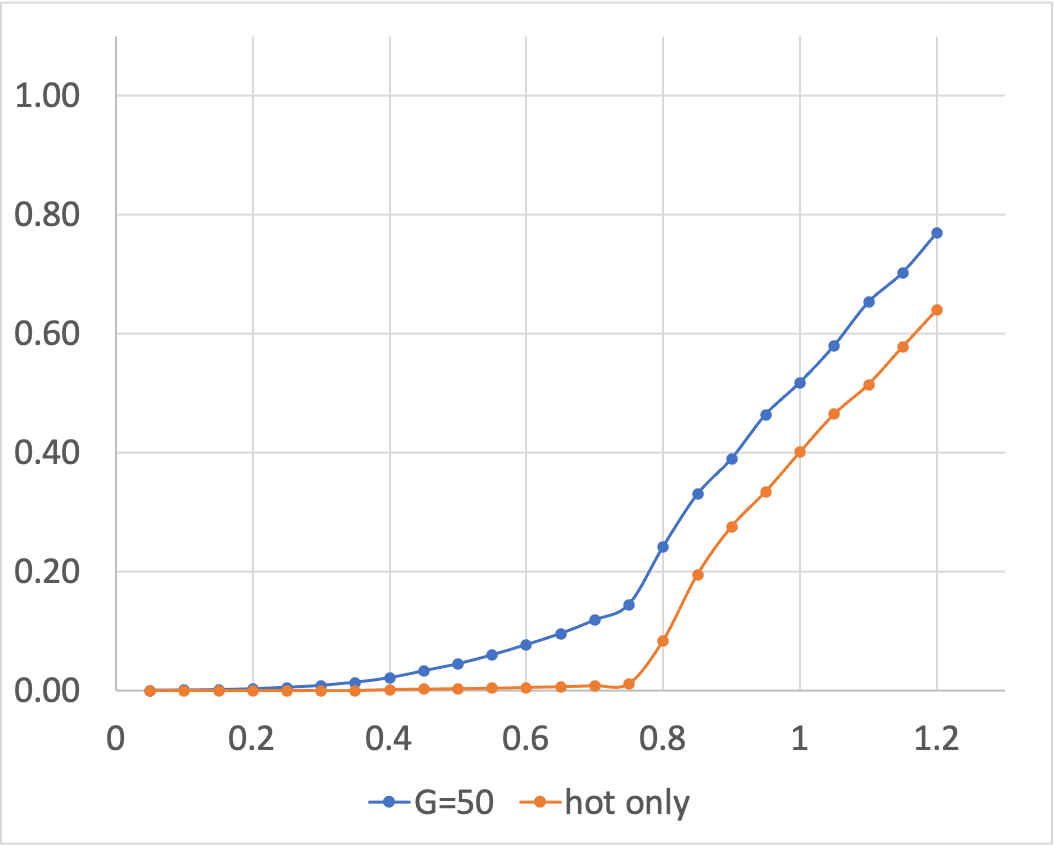}
		\caption{$G=50$}
		\label{fig:G50}
	\end{subfigure}\\
	\begin{subfigure}[b]{0.45\textwidth}
	\centering
	\includegraphics[scale=0.38]{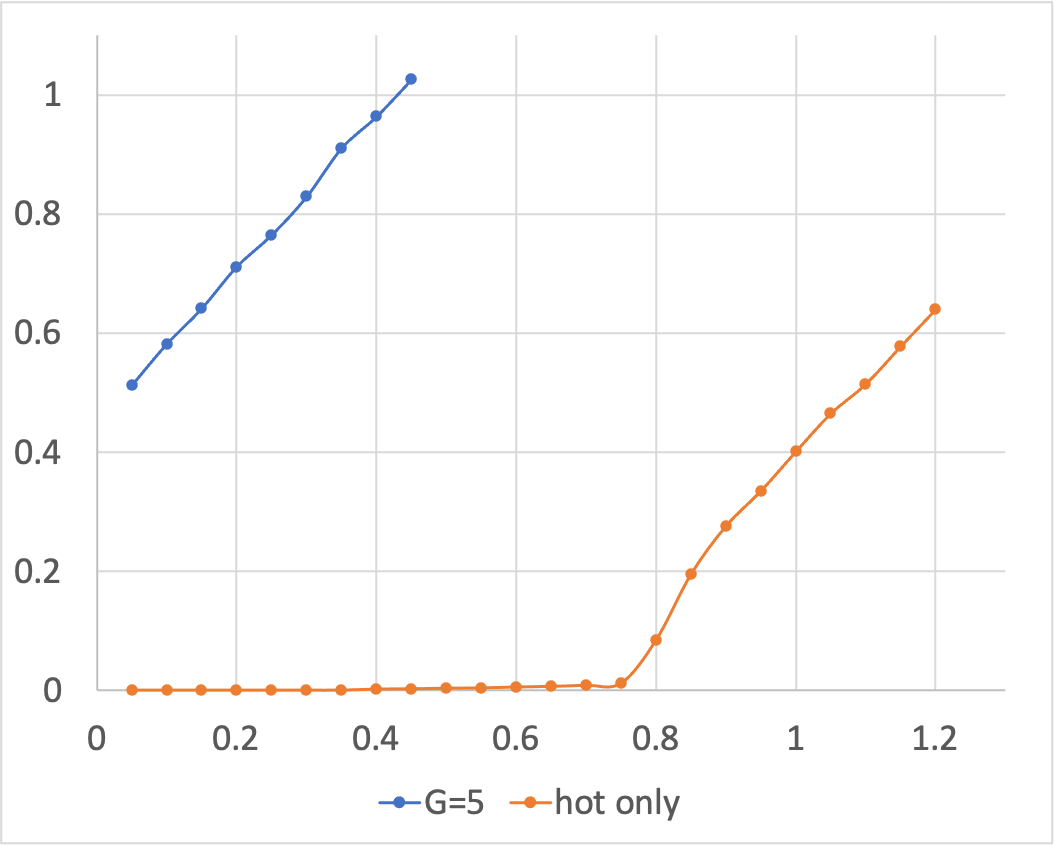}
	\caption{$G=5$}
	\label{fig:G5}
\end{subfigure}
	\caption{$\errhot_{H,I}$ for $k=3$ depending on $\lambda_h$, in presence of cold elements with $\lambda_c=5$ (blue curves) and without any cold elements (orange curve).}
	\label{fig:step-hot-error}
\end{figure}

Figure~\ref{fig:G50} illustrates that when the gap becomes larger (here, $G=50$),
the contribution of cold elements diminishes and the curve approaches the one of the uniform distribution. A larger gap leads to larger values of hot elements and, as a consequence, to a smaller relative impact of cold ones. 

Another reason for which the above conjecture may not hold is the following: even if the number of hot elements is very small but the gap factor is not large enough, the cold edges may cause the counters to become large if $\lambda_c$ is large enough, in particular in the saturation regime described in Section~\ref{sec:uni-super}. As a consequence, the ``background level'' of counters created by cold edges may be larger than true counts of hot edges, causing their overestimates. As an example, consider again the configuration  of Figure~\ref{fig:step-hot-error} 
with $k=3$ and $\lambda_c=5$. Figure~\ref{fig:uniform-supercritical} shows that the cold elements taken alone would have an error of about 6 on average ($\approx 6.25$, to be precise, data not shown) which means an about $7\times$ overestimate. Since the graph is saturated in this regime (see Section~\ref{sec:uni-super}), this means that most of the counters will be about 7 times larger than counts of cold edges. Now, if a hot element is only $5$ times more frequent than a cold one, those will be about $1.4\times$ overestimated, i.e. will have an error of about 0.4. 
% ($\approx 2.63$, to be precise). 
Referring to Figure~\ref{fig:step-hot-error}, this means that the blue curve ``starts off'' at a positive error value in this case, instead of zero. 
This situation is illustrated in Figure~\ref{fig:G5}.

\subsubsection{Mixed hypergraphs}
The analysis above shows that in presence of a ``background'' formed by large number of cold elements, the error of hot elements starts growing for much smaller load factors than without cold elements, even if the latter are much less frequent than the former. Inspired by results of Section~\ref{sec:uniform-subcritical}, one may ask if the interval of negligible error can be extended by employing the idea of variable number of hash functions. Note that here this idea applies more naturally by assigning a different number of hash functions to hot and cold elements. 

\begin{figure}[h]%{r}{0.5\textwidth}	
	\begin{center}
		\includegraphics[width=0.5\textwidth]{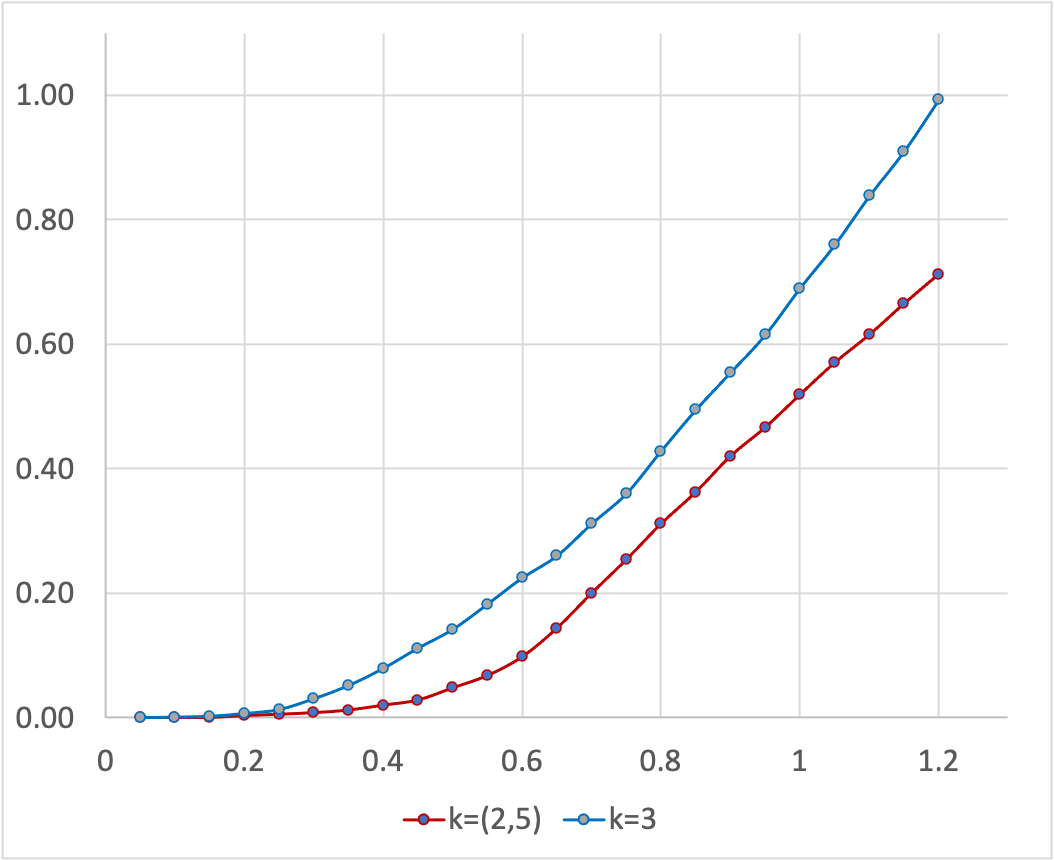}
	\end{center}
	\caption{$\errhot_{H,I}$ as a function of $\lambda_h$ for $k=3$, $\lambda_c=5$ and $G=20$ (blue curve, same as in Figure~\ref{fig:G20}) vs. $k=(2,5)$ for hot and cold elements respectively (red curve)
	}
	\label{fig:step-subcritical-mixed}
\end{figure}

Figure~\ref*{fig:step-subcritical-mixed} illustrates that this is indeed possible by assigning a smaller number of hash functions to hot elements and a larger number to cold ones. It is clearly seen that the interval supporting close-to-zero errors is extended. This happens because when the hot subgraph is not too dense, increasing the cardinality of cold edges leads to a higher probability that  at least one of the vertices of such an edge is not incident to a hot edge. As a consequence, this element does not affect the error of hot edges. For the same reason, decreasing the cardinality of hot edges (here, from 3 to 2) improves the error, as this increases the fraction of vertices non-incident to hot edges.

\subsubsection{Saturation in supercritical regime}
\label{sec:satur}
In Section~\ref{sec:uni-super} we discussed the saturation regime occurring for large load values: when the load grows sufficiently large, i.e. the hash hypergraph becomes sufficiently dense, all counters reach the same level, erasing distinctions between edges. In this regime, assuming a fixed load (graph density) and the uniform distribution of input, the edge value depends only on input size and not on the graph structure (with high probability). It is in this context that Bianchi et al. \cite{DBLP:conf/teletraffic/BianchiDLS12} studied the growth rate of edge values depending on input size. 

It is an interesting, natural and practically important question whether this saturation phenomenon holds for non-uniform distributions as well, as it is directly related to the capacity of distinguishing elements of different frequency. A full and precise answer to this question is not within the scope of this work. We believe that the answer is positive at least when the distribution is piecewise uniform, when edges are partitioned into several classes and are equiprobable within each class, provided that each class takes a linear fraction of all elements. Here we illustrate this thesis with the step distribution. 

% Beyond this interval, the error of hot edges start growing, however a question remains whether one can tell hot edges from cold ones by querying their counts. It turns out that counts of hot and cold edges significantly differ until the load reaches a certain level where their counts ``merge'' and thus, hot and cold edges become essentially indistinguishable. This is similar to the saturation phenomenon mentioned in Section~\ref{sec:uni-super}: when the hypergraph is dense enough, all its counters become ``saturated'' reaching the same level of values. 

\begin{figure}[h]
	\centering
	\begin{subfigure}[b]{0.30\textwidth}
		% \centering
		% \includegraphics[width=\textwidth]{k2.jpg}
		\includegraphics[scale=0.22]{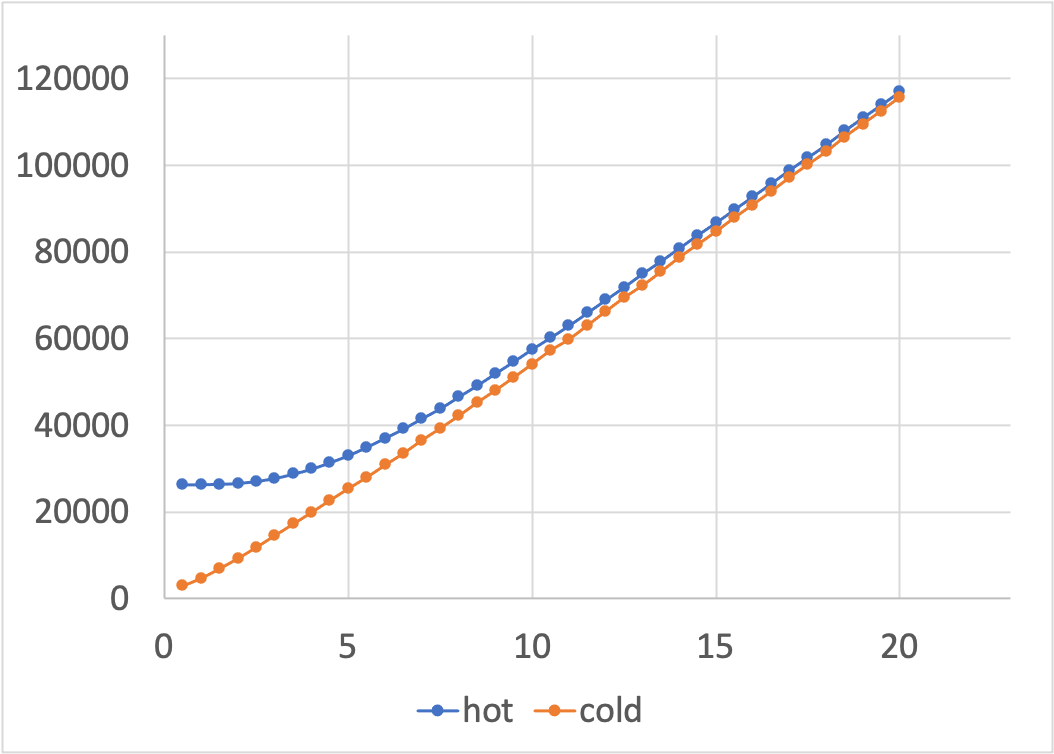}
		\caption{$2$-uniform}
		\label{fig:k2}
	\end{subfigure}
	% \hfill
	~~
	\begin{subfigure}[b]{0.30\textwidth}
		% \centering
		% \includegraphics[width=\textwidth]{k3.jpg}
		\includegraphics[scale=0.22]{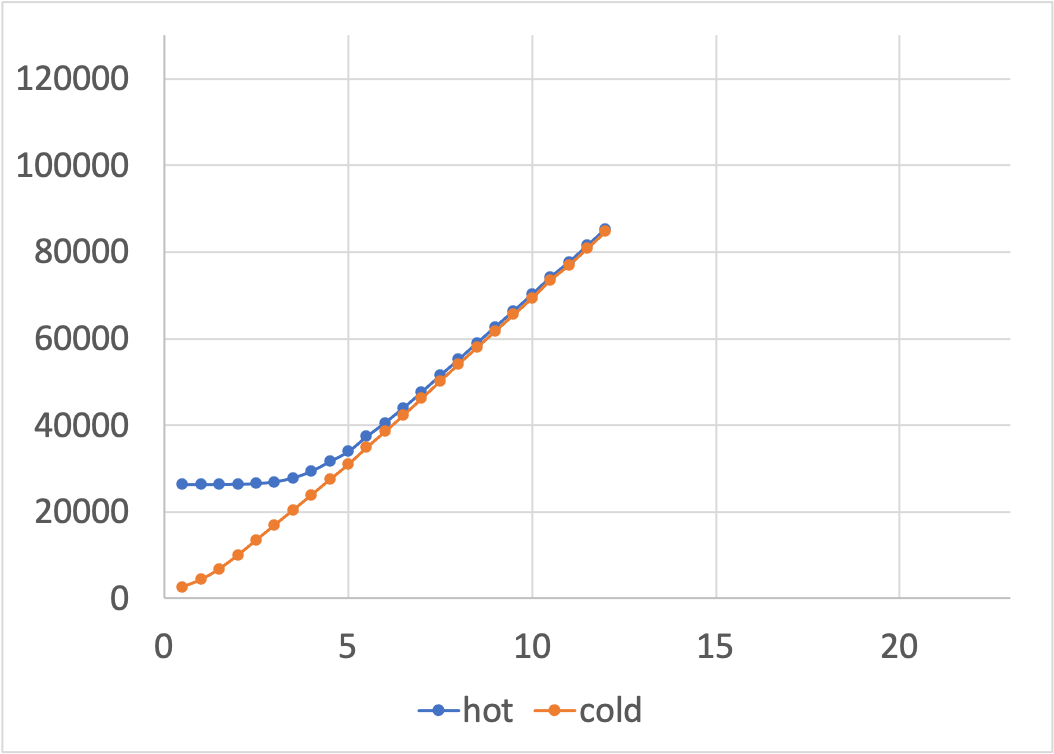}
		\caption{$3$-uniform}
		\label{fig:k3}
	\end{subfigure}
	% \hfill
	~~
	\begin{subfigure}[b]{0.30\textwidth}
		% \centering
		% \includegraphics[width=\textwidth]{k4.jpg}
		\includegraphics[scale=0.22]{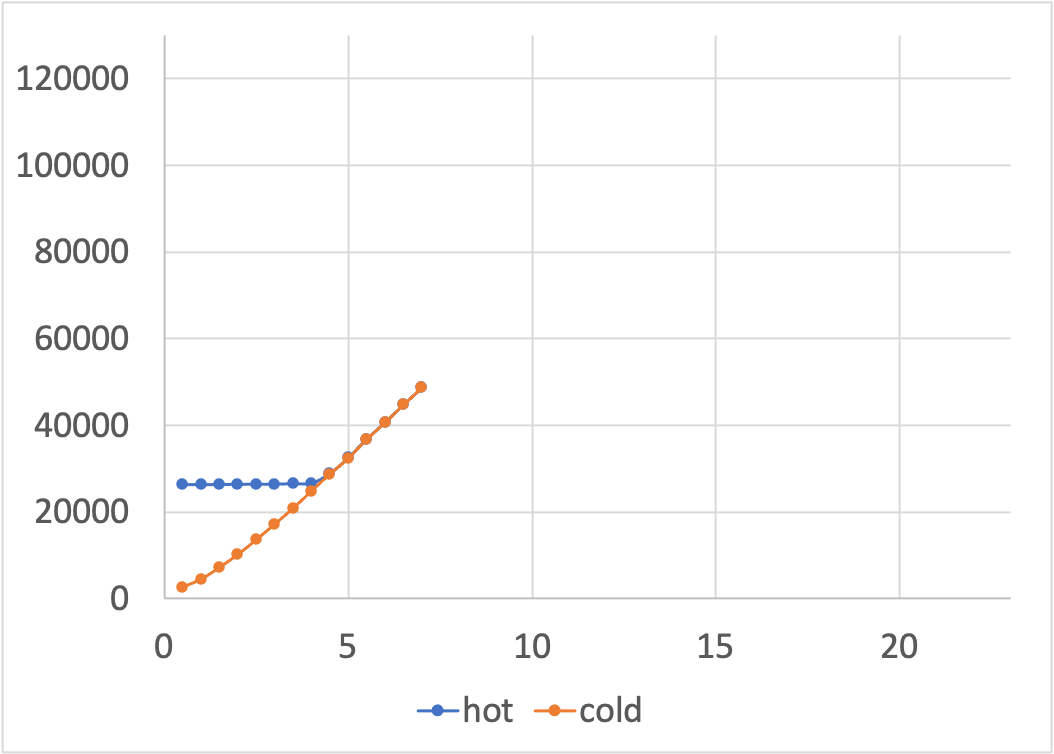}
		\caption{$(2,5)$-mixed}
		\label{fig:k2-5}
	\end{subfigure}
	\caption{Convergence of average estimates of hot and cold elements for 2-uniform (\ref{fig:k2}),  2-uniform (\ref{fig:k3}) and (2,5)-mixed (\ref{fig:k2-5}) hypergraphs. $x$-axis shows the total load $\lambda=\lambda_h + \lambda_c$ with $\lambda_h=0.1\cdot \lambda$ and $\lambda_c=0.9\cdot \lambda$ and $G=10$ in all cases. All experiences were run for $n=1000$ and the input stream of length $5,000$ times the number of edges (that is, $5\cdot 10^6\cdot \lambda$). Each average is taken over 10 random hypergraphs.}
	\label{fig:convergence}
\end{figure}

Figure~\ref{fig:convergence} illustrates the saturation phenomenon by showing average values of hot and cold edges ($G=10$) with three different configurations: 2-uniform, 3-uniform, and (2,5)-mixed. Note that the $x$-axis shows here the total load $\lambda=\lambda_h + \lambda_c$, where $\lambda_h=0.1 \cdot\lambda$ and $\lambda_c=0.9\cdot \lambda$. That is, the number of both hot and cold edges grows linearly when the total number of edges grows. 

One can observe that in all configurations, values of hot and cold edges converge, which is a demonstration of the saturation phenomenon. 
Interestingly,  the ``convergence speed'' heavily depends on the configuration: the convergence is ``slower'' for uniform configurations, whereas in the mixed configuration, it occurs right after the small error regime for hot edges.

\subsection{Zipf's distribution}
\label{sec:zipf}
Power law distributions are omnipresent in practical applications. The simplest of those is Zipf's distribution which is often used as a test case for different algorithms including Count-Min sketches \cite{CormodeMuthuICDM05,DBLP:conf/teletraffic/BianchiDLS12,DBLP:conf/iccnc/EinzigerF15,10.1145/3487552.3487856,benmazziane22b}. 
We also mention a recent learning-based variant of CountMin~\cite{hsu2019learning} (learning heavy hitters) and its study under a Zipfian distribution~\cite{aamand2019learned,du2021putting}. 
Under Zipf's distribution, element probabilities in descending order are proportional to $1/i^\beta$, where $i$ is the rank of the key and $\beta\geq 0$ is the \textit{skewness} parameter. Note that for $\beta = 0$, Zipf's distribution reduces to the uniform one. 
% It is therefore a natural question whether the phase transition occurs for Zipf's distributions with $\beta > 0$. 

%If one considers the combined error, it has been observed in our previous work \cite{FusyKucherovCIAC23} that the phase transition phenomenon ceases to hold. Similar to the step distribution, one can ask if assigning a different number of hash functions can make the error smaller. More specifically, can it help if frequent (``hot'') elements are assigned a smaller number hash functions than less frequent (``cold'') ones? 
%
%Figure~\ref{fig:zipf07-subcritical} shows the combined error for uniform ($k=2,3$) and mixed graphs, for small values of $\lambda$. In the latter, most frequent edges (40\% of all) are of carnality 2 and the remaining ones of cardinality 5. We observe a significant reduction of error caused by the application of mixed graphs. 
% 
%\begin{figure}[h]%{r}{0.5\textwidth}	
%	\begin{center}
%		\includegraphics[width=0.7\textwidth]{zipf07-subcritical.png}
%	\end{center}
%	\caption{Combined error for Zipf's distribution ($\beta=0.7$) for 2-uniform, 3-uniform and (2,5)-mixed hypergraphs, for small $\lambda$.
%	}
%	\label{fig:zipf07-subcritical}
%\end{figure}

Zipf's distribution is an important test case for our study as well, as it forces several (few) most frequent elements to have very large counts and a large number of elements (\textit{heavy tail}) to have small counts whose values decrease only polynomially on the element rank and are therefore of the same order of magnitude. Bianchi et al. \cite[Fig.~1]{DBLP:conf/teletraffic/BianchiDLS12} observed that for Zipf's distribution in the supercritical regime, the estimates follow the ``waterfall-type behavior'': the most frequent elements have essentially exact estimates whereas the other elements have all about the same estimate regardless of their frequency. Figure~\ref{fig:zipf-waterfall} illustrates this phenomenon for different skewness values. 

\begin{figure}[h]
	\centering
	\begin{subfigure}[b]{0.32\textwidth}
		\centering
		\includegraphics[scale=0.25]{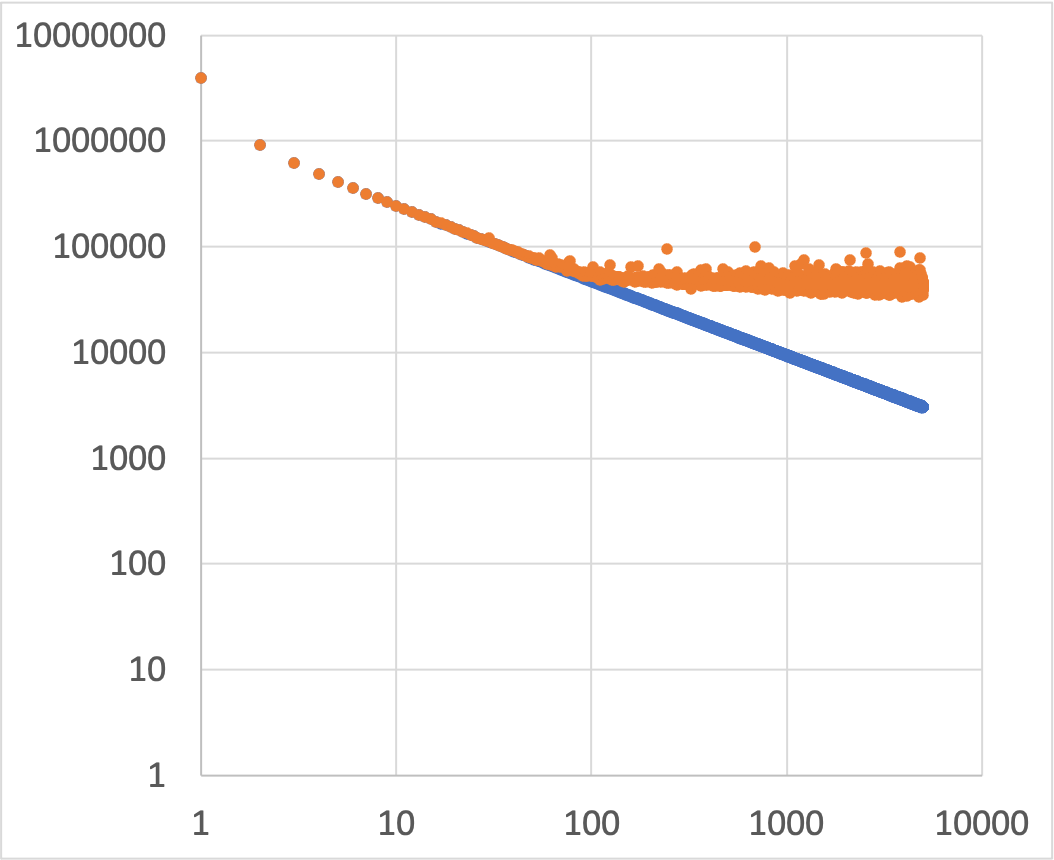}
		\caption{$\beta=0.7$}
		\label{fig:zipf-beta07}
	\end{subfigure}
	% \hspace*{10pt}
	\begin{subfigure}[b]{0.32\textwidth}
		\centering
		\includegraphics[scale=0.25]{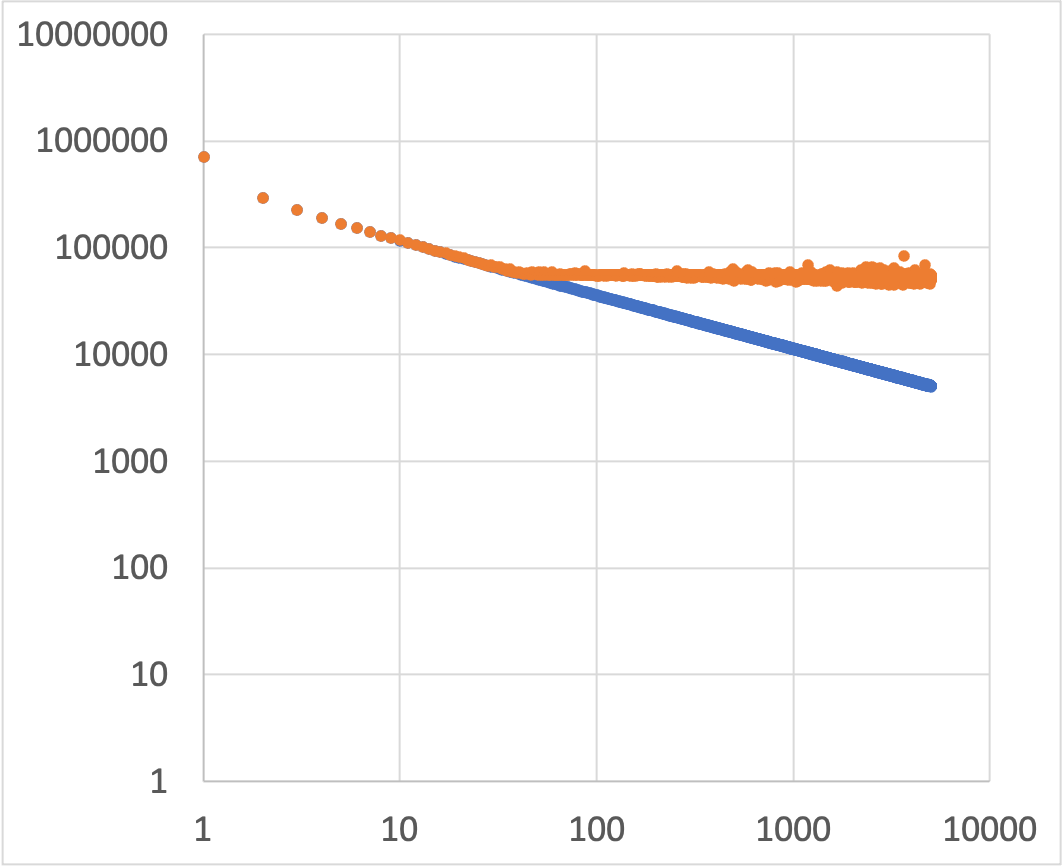}
		\caption{$\beta=0.5$}
		\label{fig:zipf-beta05}
	\end{subfigure}
	\begin{subfigure}[b]{0.32\textwidth}
		\centering
		\includegraphics[scale=0.25]{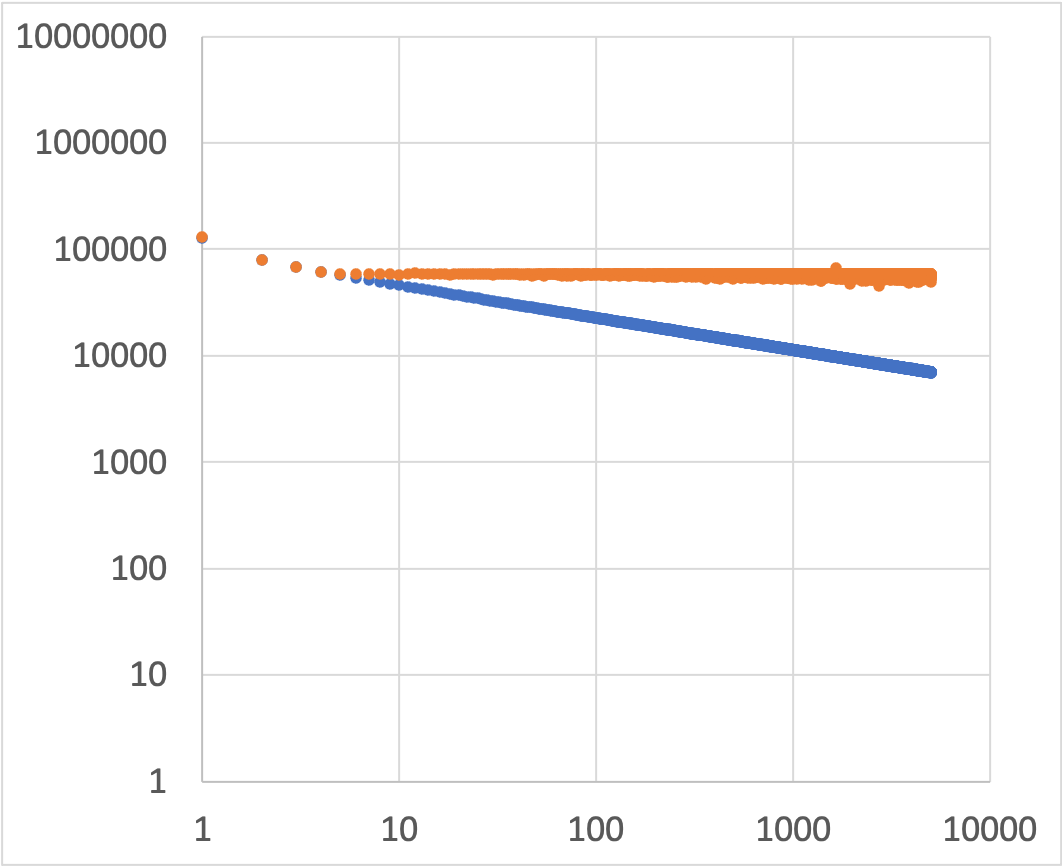}
		\caption{$\beta=0.3$}
		\label{fig:zipf-beta03}
	\end{subfigure}
	\caption{Exact (blue) and estimated (orange) edge values for Zipf's distribution as a function on the element frequency rank, plotted in double log scale. All plots obtained for $n=1000$, $\lambda=5$, $k=2$, and the input size $50\cdot 10^6$. Estimates are averaged over 10 hash function draws.}
	\label{fig:zipf-waterfall}
\end{figure}

The waterfall-type behavior for Zipf's distribution is well explained by the analysis we developed in the previous sections. The ``waterfall pool level'' of values (called \textit{error floor} in \cite{DBLP:conf/teletraffic/BianchiDLS12}) is the effect of saturation formed by heavy tail elements. The few ``exceptionally frequent'' elements are too few to affect the saturation level (their number is $\ll n$), they turn out to constitute ``peaks'' above the level and are thus estimated without error. 
Naturally, smaller skewness values make the distribution less steep and reduce the number of ``exceptionally frequent'' elements. For example, according to Figure~\ref{fig:zipf-waterfall}, for $\lambda=5$ and $k=2$, about 50 most frequent elements are evaluated without error for $\beta=0.7$, about 40 for $\beta=0.5$ and only 5 for $\beta=0.3$. % [PEUT-ON QUANTIFIER CELA??? CE SERAIT TRES INTERESSANT]

\begin{figure}[h]
	\centering
	\begin{subfigure}[b]{0.45\textwidth}
		\centering
		\includegraphics[scale=0.25]{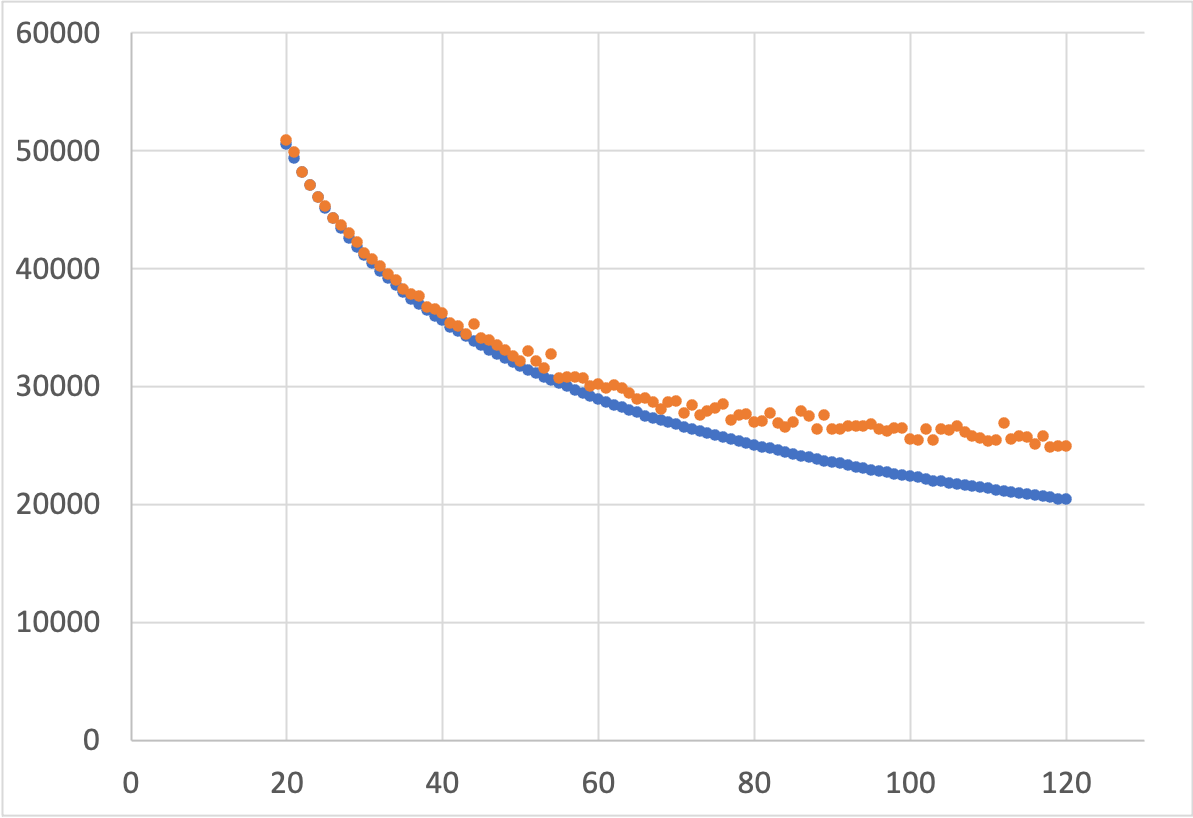}
		\caption{$k=2$}
		\label{fig:zipf-zoom-k2}
	\end{subfigure}
	% \hspace*{10pt}
	~~
	\begin{subfigure}[b]{0.45\textwidth}
		\centering
		\includegraphics[scale=0.25]{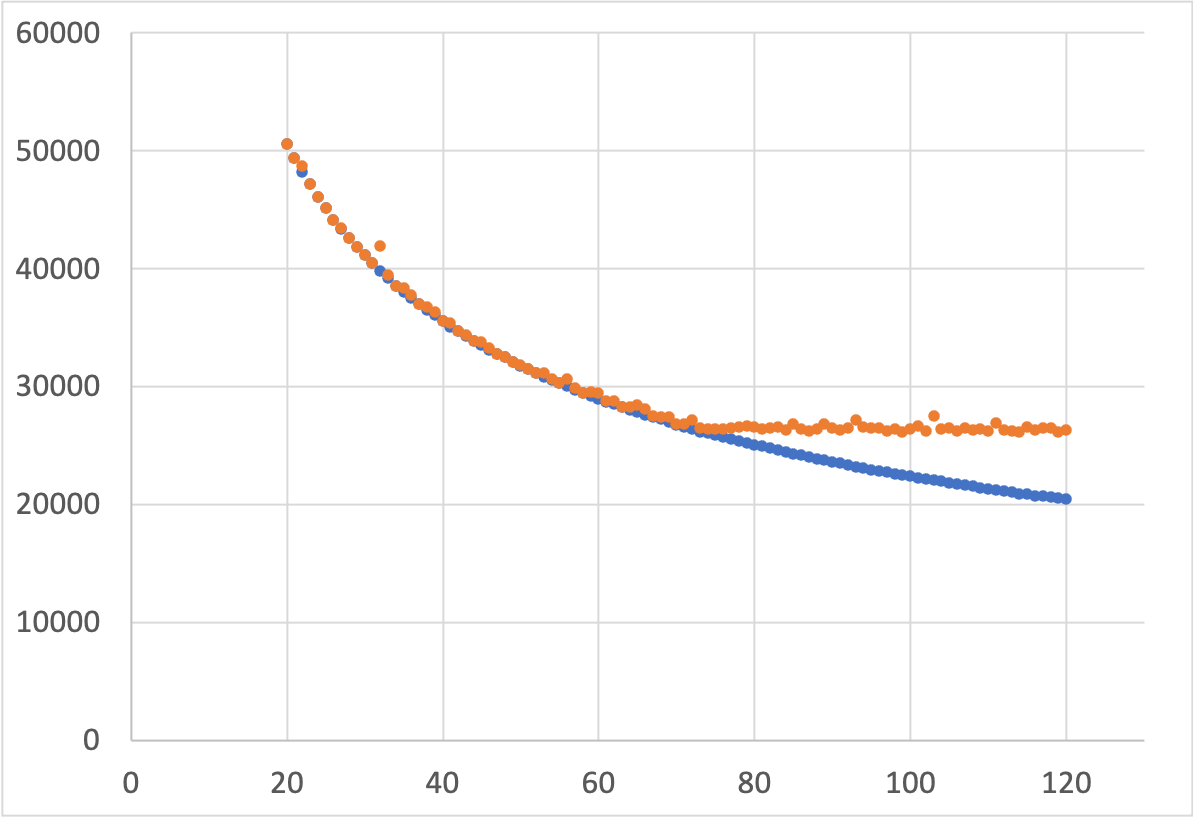}
		\caption{$k=(2,5;0.2)$}
		\label{fig:zipf-zoom-k2-5}
	\end{subfigure}
	\caption{Exact (blue) and estimated (orange) edge values for Zipf's distribution as a function on the element frequency rank, for $n=1000$, $\lambda=2$, and input size $20\cdot 10^6$. Values for ranks 20 to 120 only are shown. Estimates are averaged over 50 hash function draws.}
	\label{fig:zipf-waterfall-zoom}
\end{figure}

Following our results from previous sections, we studied whether using a varying number of hash functions can extend the range of frequent elements estimated with small error. We found that for moderate loads $\lambda$, this is possible indeed. Figure~\ref{fig:zipf-waterfall-zoom}  illustrates this for $\lambda=2$. It shows a ``zoom'' around the ``break point'' (see Figure~\ref{fig:zipf-waterfall}) for $k=2$ vs. $k=(2,5;0.2)$. For clarity, plots are shown in regular scale and for elements of rank 20 to 120 only. The Figure demonstrates that the case $k=(2,5;0.2)$ provides a sharper break (see also Figure~\ref{fig:convergence}). As a result, even if the saturation level is higher for $k=(2,5;0.2)$ than for $k=2$, about 70 most frequent elements are evaluated with small error with $k=(2,5;0.2)$, vs. about 50 with $k=2$. 

\section{Conclusions}

In this paper, we presented a series of experimental results providing new insights into the behavior of conservative Count-Min sketch. Some of them have direct applications to practical usage of this data structure. Main results can be summarized as follows. 
\begin{itemize}
	\item For the uniform distribution of input elements, assigning a different number of hash functions to different elements extends the subcritical regime (range of load factors $\lambda$) that supports asymptotically vanishing relative error. This immediately implies space saving for Count-Min configurations verifying this regime. For non-uniform distributions, varying number of hash functions allows extending the regime of negligible error for most frequent elements, 
%	\item In a supercritical regime (here, $\lambda > 6$) and uniform distribution, for a limited range of load factors (roughly, $\lambda < 50$), varying number of hash functions allows reducing the error even compared to the configuration $k=1$ which yields the smallest error for constant $k$'s,
	\item Under ``sufficiently uniform distributions'', including uniform and step distributions, a Count-Min sketch reaches a saturation regime when $\lambda$ becomes sufficiently large. In this regime, counters become concentrated around the same value and elements with different frequency become indistinguishable,
	\item Frequent elements that can be estimated with small error can be seen as those which surpass the saturation level formed by the majority of other elements. For example, in case of Zipf's distribution, those elements are a few ``exceptionally frequent elements'', whereas the saturation is insured by the heavy-tail elements.  Applying a varying number of hash functions can increase the number of those elements for moderate loads $\lambda$. 
\end{itemize}

Many of those results lack a precise mathematical analysis. Perhaps the most relevant to practical usage of Count-Min is the question of estimating the ``waterfall pool level'' for different distributions, 
% as it provides a lower bound to the frequency of elements that will be estimated with small error, 
which is a fundamental information for heavy-hitter type of applications. Indeed, in a ``very supercritical'' regime, estimates tend to converge to the same level, as it was illustrated in Section~\ref{sec:satur} for the step distribution. On the other hand, a few elements with ``abnormally high'' frequencies that exceed this level are those which are evaluated with small error. This is illustrated with Zipf's distribution (Section~\ref{sec:zipf}). 
Bianchi et al. \cite{DBLP:conf/teletraffic/BianchiDLS12} observed that in the case of non-uniform distributions of input elements, the ``waterfall pool level'' is upper-bounded by the saturation level of counter values when the input distribution is modeled by a uniform choice among all $\binom{n}{k}$ possible edges. 
% for the uniform distribution of input. 
This latter is computed in \cite{DBLP:conf/teletraffic/BianchiDLS12} using a method based on Markov chains and differential equations. We believe that this method can be extended to the case of mixed graphs as well and leave it for future work. % Evaluating this level is important as it determines the level of counter values such that elements whose frequency exceeds this level are evaluated with a negligible error. 
However, computing the ``waterfall pool level'' for non-uniform distributions including Zipf's distribution is an open problem. 

%Regarding an analysis, in the $k$-uniform case, following~\cite{DBLP:conf/teletraffic/BianchiDLS12} one can compute (based on a system of differential equations) a theoretical common level $L$ reached by the counters
%under an input stream of same length (and with same number of counters) as the one considered, under the model where the counters for each input element are uniformly and 
%independently chosen among the $\binom{n}{k}$ possibilities. 
%As supported experimentally in~\cite{DBLP:conf/teletraffic/BianchiDLS12}, it is to be expected that the average estimates for hot elements will be above the average estimate for cold elements whenever the common multiplicity of hot elements is above $L$. The same should hold in the mixed case, upon extending the approach of~\cite{DBLP:conf/teletraffic/BianchiDLS12} to
%compute the theoretical level $L$ for a mixed input stream (where a fraction $\alpha$ of input elements have multiplicity $k_1$ while the remaining fraction has multiplicity $k_2$, 
%and the counters hit by an input element are chosen uniformly and independently from the other elements): the system of differential equations leading to $L$ would be a mixture of the 
%ones for the $k_1$ uniform case and the ones for the $k_2$-uniform case. 

% Interestingly, more frequently queried keys in Weighted Bloom filters are assigned more hash functions \cite{DBLP:conf/isit/BruckGJ06}, whereas more frequent keys in Weighted Count-Min sketches are, in the contrary, assigned less hash functions(?)

% \bibliographystyle{alpha}
\bibliography{biblio2.bib}

\end{document}